\documentclass[10pt,conference]{llncs}

\usepackage[noadjust]{cite}

\usepackage{cmap}
\usepackage{amsmath}
\usepackage{listings}
\usepackage{amsfonts}
\usepackage{graphicx}
\usepackage{courier}
\usepackage{algorithm}
\usepackage[noend]{algpseudocode}
\let\oldReturn\Return
\renewcommand{\Return}{\State\oldReturn}
\usepackage[table,xcdraw]{xcolor}
\usepackage{float}
\usepackage{hyperref}
\usepackage{mathtools}
\usepackage[framemethod=TikZ]{mdframed}
\usepackage[]{inputenc}
\usepackage[T1]{fontenc}
\usepackage{placeins}
\usepackage{amsmath,amssymb}
\usepackage{xspace}
\usepackage[font=footnotesize]{caption}
\usepackage[framemethod=TikZ]{mdframed}
\usepackage{booktabs}
\usepackage{fancyvrb}
\usepackage{relsize}
\graphicspath{{images/}}
\usepackage{float}
\usepackage{subfig}
\usepackage{framed}
\usepackage{multirow}

\usepackage[firstpage]{draftwatermark}
\SetWatermarkText{\hspace*{6in}\raisebox{6.3in}{\includegraphics[scale=0.1]{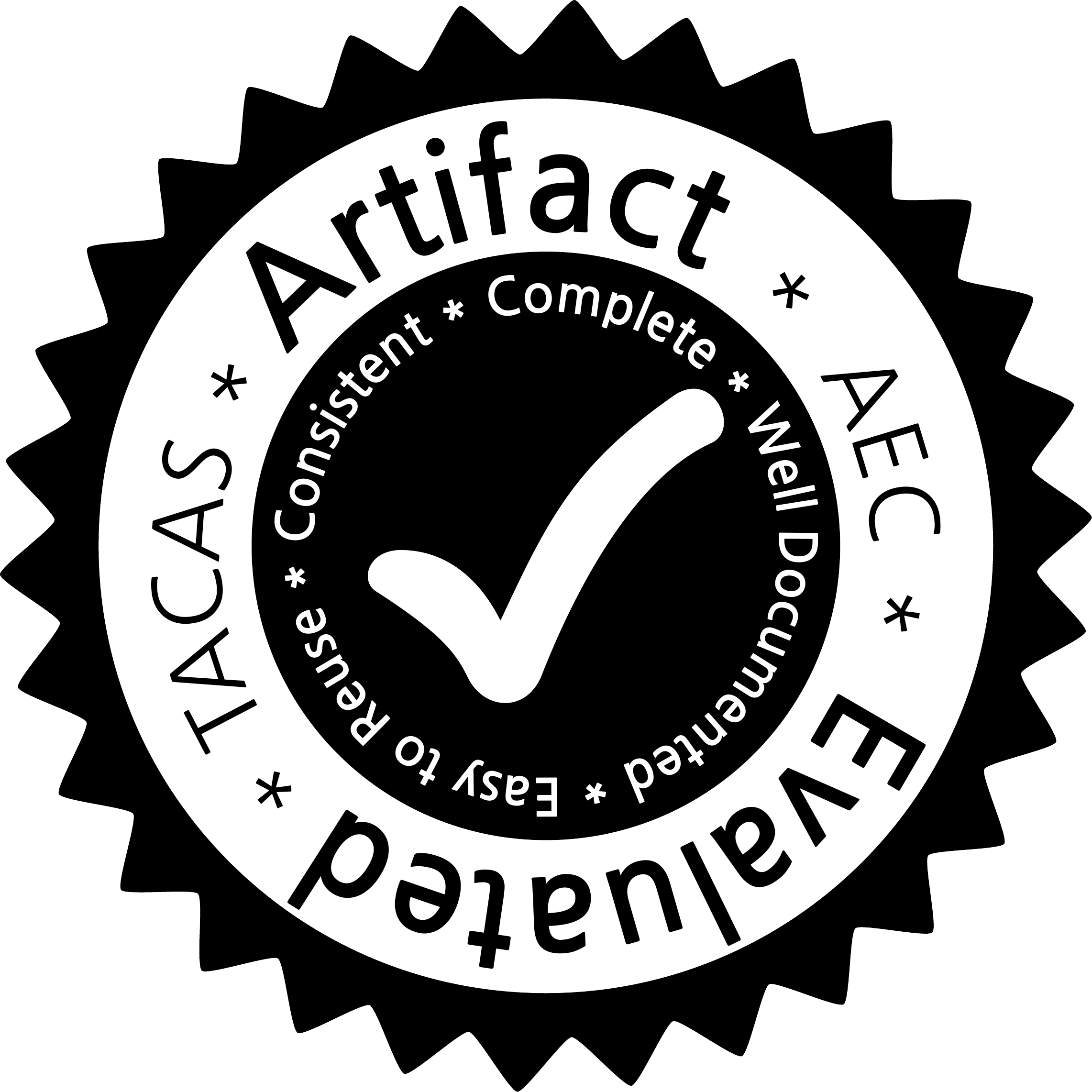}}}
\SetWatermarkAngle{0}

\newcommand{\aeval}{\textsc{AE-VAL}\xspace}
\newcommand{\jkind}{\textsc{JKind}\xspace}

\newcommand{\jsyn}{\textsc{JSyn}\xspace}

\newcommand{\jsynvg}{\textsc{JSyn-vg}\xspace}
\newcommand{\smtlibtoc}{\textsc{SMTLib2C}\xspace}

\newcommand{\viable}{{\mathsf {Viable}}}

\newcommand{\glb}{\textit {GLB}\xspace}
\newcommand{\lub}{\textit {LUB}\xspace}
\newcommand{\tuple}[1]{\langle #1 \rangle}

\newcommand{\andrew}[1]{\textcolor{green}{Andrew: #1}}
\newcommand{\john}[1]{\textcolor{orange}{John: #1}}
\newcommand{\grigory}[1]{\textcolor{brown}{Grigory: #1}}

\newcommand{\realizable}{\textsc{realizable}\xspace}
\newcommand{\unrealizable}{\textsc{unrealizable}\xspace}
\newcommand{\skolems}{\textit{Skolem}}

\newcommand{\subs}{\textit{validRegion}}

\newcommand\eqdef{\mathrel{\stackrel{\makebox[0pt]{\mbox{\normalfont\tiny def}}}{=}}}

\newcounter{template}

\begin{document}
\title{Validity-Guided Synthesis of Reactive Systems from Assume-Guarantee Contracts}

\author{Andreas Katis\inst{1}, Grigory Fedyukovich\inst{2}, Huajun Guo\inst{1},
  Andrew Gacek\inst{3}, John Backes\inst{3}, Arie Gurfinkel\inst{4}, Michael
  W. Whalen\inst{1}}%
 
\institute{
Department of Computer Science and Engineering, University of Minnesota\\
\email{\{katis001,guoxx663\}@umn.edu, whalen@cs.umn.edu}
\and
Department of Computer Science, Princeton University\\
\email{grigoryf@cs.princeton.edu}
\and
Rockwell Collins Advanced Technology Center\\
\email{\{andrew.gacek,john.backes\}@rockwellcollins.com}
\and
Department of Electrical and Computer Engineering, University of Waterloo\\
\email{agurfinkel@uwaterloo.ca}
}


\maketitle

\begin{abstract}

Automated synthesis of reactive systems from specifications has been a topic of research for decades.  Recently, a variety of approaches have been proposed to extend synthesis of reactive systems from propositional specifications towards specifications over rich theories.
We propose a novel, completely automated approach to program synthesis which reduces the problem to deciding the validity of a set of $\forall\exists$-formulas.
In spirit of
IC3 / PDR, our problem space is recursively refined by blocking out regions of unsafe states, aiming to discover a fixpoint that describes safe reactions.
If such a fixpoint is found, we construct a witness that is directly translated into an implementation.
We implemented the algorithm on top of the \jkind model checker, and exercised it against contracts written using
the Lustre specification language.
Experimental results show how the new algorithm outperforms
\jkind's already existing synthesis procedure based on $k$-induction and addresses soundness issues in the $k$-inductive approach with respect to unrealizable results.
\end{abstract}

\section{Introduction}

Program synthesis is one of the most challenging problems in computer science. The objective is to define a process to automatically derive implementations that are guaranteed to comply with specifications expressed in the form of logic formulas. The problem has seen increased popularity in the recent years, mainly due
to the capabilities of modern symbolic solvers, including Satisfiability Modulo Theories (SMT)~\cite{BarFT-SMTLIB} tools, to compute compact and precise regions that describe under which conditions an implementation exists for the given specification~\cite{reynolds2015counterexample}.
 As a result, the problem has been well-studied for the area of propositional specifications (see Gulwani~\cite{gulwani2010dimensions} for a survey), and approaches have been proposed to tackle challenges involving richer specifications. Template-based techniques focus on synthesizing programs that
match a certain shape (the template)~\cite{srivastava2013template}, while {\em inductive synthesis} uses the idea of refining the problem space using counterexamples, to converge to a solution~\cite{flener2001inductive}.
A different category is that of \textit{functional synthesis}, in which the goal is to construct functions from pre-defined input/output relations~\cite{kuncak2013functional}.

Our goal is to effectively synthesize programs from safety specifications written in the Lustre~\cite{lustrev6} language.  These specifications are structured in the form of {\em
Assume-Guarantee} contracts, similarly to approaches in Linear Temporal Logic~\cite{ringert2017synthesis}. In prior work, we developed a solution to the synthesis problem which is based on $k$-induction~\cite{gacek2015towards,katis2016towards,KatisFGBGW16}.
Despite showing good results, the approach suffers from soundness problems with respect to unrealizable results; a contract could be declared as unrealizable, while an actual implementation exists.
In this work, we propose a novel approach that is a direct improvement over the $k$-inductive method in two important aspects: performance and generality.  On all models that can be synthesized
by $k$-induction, the new algorithm always outperforms in terms of synthesis time while yielding roughly approximate code sizes and execution times for the generated code. More importantly, the new algorithm can synthesize a strictly larger set of benchmark models,
and comes with an improved termination guarantee: unlike in $k$-induction, if the algorithm terminates with an ``unrealizable'' result, then there is no possible realization of the contract.

The technique has been used to synthesize contracts involving linear real and integer arithmetic (LIRA),
but remains generic enough to be extended into supporting additional theories
in the future, as well as to liveness properties that can be reduced to safety properties (as in $k$-liveness~\cite{claessen2012liveness}).  Our approach is completely automated and requires no guidance to the tools in terms of user interaction (unlike~\cite{ryzhyk2014user,ryzhyk2016developing}), and it is capable of providing solutions without requiring any templates, as in e.g., work by Beyene et. al.~\cite{beyene2014constraint}.  We were able to automatically solve problems that were ``hard'' and required hand-written templates specialized to the problem in~\cite{beyene2014constraint}.

The main idea of the algorithm was inspired by induction-based model checking, and in particular by IC3 / Property Directed Reachability (PDR)~\cite{bradley2011sat,een2011efficient}. In PDR, the goal is to discover an inductive invariant for a property, by recursively blocking generalized regions describing unsafe states. Similarly, we attempt
to reach a greatest fixpoint that contains states that react to arbitrary environment behavior and lead to states within the fixpoint that comply with all guarantees. Formally, the greatest fixpoint is sufficient to prove the validity of a $\forall\exists$-formula, which states that for any state and environment input, there exists a system reaction that complies with the specification. Starting from the entire
problem space, we recursively block regions of states that violate the contract, using \textit{regions of validity} that are
generated by invalid $\forall\exists$-formulas.
If the refined
$\forall\exists$-formula is valid, we reach a fixpoint which can effectively be used by the specified transition relation to
provide safe reactions to environment inputs. We then extract a witness for the
formula's satisfiability%
, which can be directly transformed into the
language intended for the system's implementation.

The algorithm was implemented as a feature in the \jkind model checker and is based on the general
concept of extracting a witness that satisfies a $\forall\exists$-formula, using
the \aeval Skolemizer~\cite{fedyukovich2015automated,KatisFGBGW16}. While \aeval was mainly used as a tool for solving queries and extracting Skolems in our $k$-inductive approach, in this paper we also take advantage of its capability to generate
\textit{regions of validity} from invalid formulas to reach a fixpoint of satisfiable assignments to state variables.

The contributions of the paper are therefore:
\begin{itemize}
    \item A novel approach to synthesis of contracts involving rich theories that is efficient, general, and completely automated (no reliance on templates or user guidance),
    \item an implementation of the approach in a branch of the \jkind model checker, and
    \item an experiment over a large suite of benchmark models demonstrating the effectiveness of the approach.
\end{itemize}

The rest of the paper is organized as follows. Sect.~\ref{sec:example} briefly describes the Cinderella-Stepmother problem that we use as an example throughout the paper. In Sect.~\ref{sec:background}, we provide the necessary formal definitions to describe the synthesis algorithm, which is presented then in Sect.~\ref{sec:synthesis}.
We present an evaluation in Sect.~\ref{sec:impl} and comparison against a method based on $k$-induction that exists using the same input language.
Finally, we discuss the differences of our work with closely related ideas in Sect.~\ref{sec:related} and conclude in Sect.~\ref{sec:conclusion}.


\section{Overview: The Cinderella-Stepmother Game}
\label{sec:example}

We illustrate the flow of the validity guided-synthesis algorithm using a variation of the minimum-backlog
problem, the two player game between Cinderella and her wicked Stepmother, first expressed by Bodlaender \textit{et al.}~\cite{bodlaender2012cinderella}.

The main objective for Cinderella (i.e. the reactive system) is to prevent a
collection of buckets from overflowing with water. On the other hand,
Cinderella's Stepmother (i.e. the system's environment) refills the buckets with a predefined amount of water that is distributed in a random fashion between the buckets.
For the running example, we chose an instance of the game that has been
previously used in template-based synthesis~\cite{beyene2014constraint}. In this instance, the game is described
using five buckets, where each bucket can contain up to two units of water.
Cinderella has the option to empty two adjacent buckets at each of her turns,
while the Stepmother distributes one unit of water over all five buckets. In the context of this paper we use this example to show how specification is expressed, as well as how we can synthesize an efficient implementation that describes reactions for Cinderella, such that a bucket overflow is always prevented.

\begin{figure}[!t]
\centering
\includegraphics[scale=0.6]{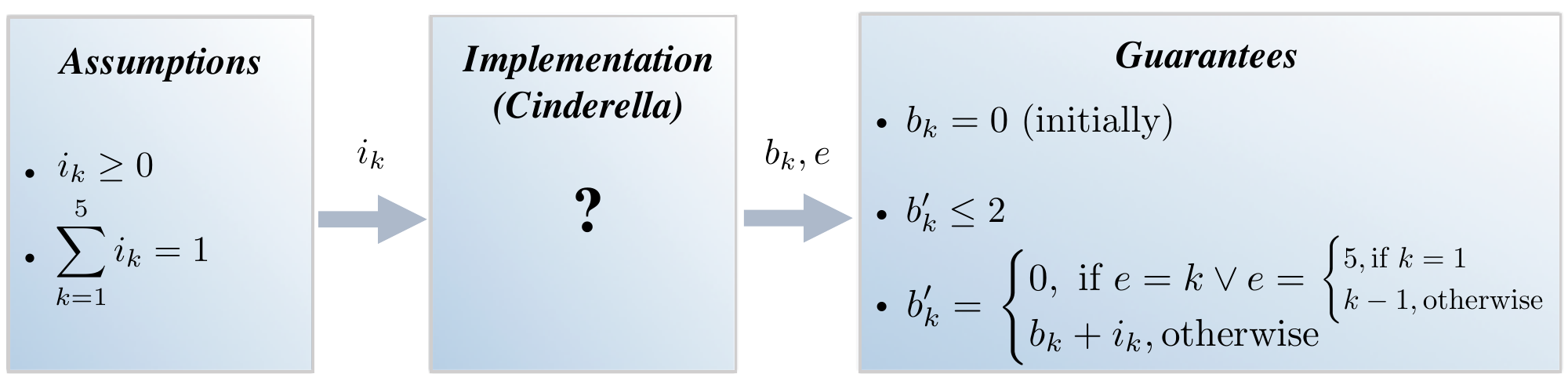}
\caption{An Assume-Guarantee contract.}
\label{fg:agcontract}
\end{figure}

We represent the system requirements using an \textit{Assume-Guarantee
Contract}. The \emph{assumptions} of the contract restrict the possible inputs that the
environment can provide to the system, while the \emph{guarantees}
describe safe reactions of the system to the outside world.

A (conceptually) simple example is shown in Fig.~\ref{fg:agcontract}. The contract describes a possible set of requirements for a specific instance of the Cinderella-Stepmother game. 
Our goal is to synthesize an implementation that describes Cinderella's winning region of the game. Cinderella in this case is the implementation, as shown by the middle box in Fig.~\ref{fg:agcontract}. Cinderella's inputs are five different values $i_k$, $1 \leq k \leq 5$, determined by a random distribution of one unit of water by the Stepmother. During each of her turns Cinderella has to make a choice denoted by the output variable $e$, such that the buckets $b_k$ do not overflow during the next action of her Stepmother. We define the contract using the set of assumptions $A$ (left box in Fig.~\ref{fg:agcontract}) and the guarantee constraints $G$ (right box in Fig.~\ref{fg:agcontract}). For the particular example, it is possible to construct at least one implementation that satisfies $G$ given $A$ which is described in Sect.~\ref{sec:algexample}.
The proof of existence of such an implementation is the main concept behind the \emph{realizability} problem, while the automated construction of a witness implementation is the main focus of \emph{program synthesis}.

Given a proof of realizability of the contract in Fig.~\ref{fg:agcontract}, 
we are seeking for an efficient synthesis procedure that could provide an implementation.
On the other hand, consider a variation of the example, where $A = \mathit{true}$. This is a practical case of an
\emph{unrealizable} contract, as there is no feasible Cinderella implementation that can correctly react to Stepmother's actions. An example counterexample allows the Stepmother to pour random amounts of water into the buckets, leading to overflow of at least one bucket during each of her turns.

\section{Background}
\label{sec:background}
We use two disjoint sets, $state$ and $inputs$, to describe a system.
A straightforward and intuitive way to represent an \emph{implementation} is by
defining a \emph{transition system}, composed of an initial state
predicate $I(s)$ of type $state \to bool$, as well as a transition relation
$T(s,i,s')$ of type $state \to inputs \to state \to bool$.

Combining the above, we represent an Assume-Guarantee (AG) contract using a set
of \emph{assumptions}, $A: state \rightarrow inputs \rightarrow bool$,
and a set of \emph{guarantees} $G$. The latter is further decomposed into two
distinct subsets $G_I: state \rightarrow bool$ and $G_T: state \rightarrow
inputs \rightarrow state \rightarrow bool$. The $G_I$ defines the set of valid
initial states, and $G_T$ contains constraints that need to be satisfied in
every transition between two states. Importantly, we
do not make any distinction between the internal state variables and the output variables in the
formalism. This allows us to use the state variables to (in some cases)
simplify the specification of guarantees since a contract
might not be always defined over all variables in the transition system.

Consequently, we can formally define a realizable contract, as one for which any
preceding state $s$ can  transition into a new state $s'$ that satisfies
the guarantees, assuming valid inputs. For a system to be ever-reactive, these
new states $s'$ should be further usable as preceding states in a future
transition. States like $s$ and $s'$ are called \textit{viable} if
and only if:
\begin{align}
\begin{split}
  \viable(s) &=
  \forall i. (A(s, i) \Rightarrow \exists s'.~ G_T(s, i,s')
\land \viable(s'))
\label{eq:viable}
\end{split}
\end{align}
This equation is recursive and we interpret it coinductively, i.e., as a
greatest fixpoint.
A necessary condition, finally, is that the intersection of sets of viable states
and initial states is non-empty. As such, to conclude that a contract
is realizable, we require that
\begin{equation}
\exists s. G_I(s) \land \viable(s)
\label{eq:nonempty}
\end{equation}

\noindent The synthesis problem is therefore to determine an initial state $s_i$ and function $f(s, i)$ such that $G_I(s_i)$ and $\forall s, i . \viable(s) \Rightarrow \viable(f(s, i))$.

The intuition behind our proposed algorithm in this paper relies on the
discovery of a fixpoint $F$ that only contains viable states.  We can determine whether $F$ is a fixpoint by proving the validity of the following formula:
\[
\forall s,i. \ (F(s) \land A(s,i) \Rightarrow \exists s'.G_{T}(s,i,s') \land F(s'))
\]

\noindent In the case where the greatest fixpoint $F$ is non-empty, we check whether it satisfies $G_{I}$ for some initial state.  If so, we proceed by extracting a witnessing initial state and witnessing skolem function $f(s, i)$ to determine $s'$ that is, by construction, guaranteed to satisfy the specification.

To achieve both the fixpoint generation and the witness extraction, we depend on \aeval, a solver for $\forall\exists$-formulas.

\subsection{Skolem functions and regions of validity}
\label{sec:aeval}


We rely on the already established algorithm to decide the validity of $\forall\exists$-formulas and extract Skolem functions, called \aeval~\cite{fedyukovich2015automated}.
It takes as input a formula $\forall x \,.\, \exists y  \,.\, \Phi (x, y)$ where $\Phi (x, y)$ is quantifier-free.
To decide its validity, \aeval first normalizes $\Phi (x, y)$ to the form $S(x) \Rightarrow T(x, y)$ and then attempts to extend all models of $S(x)$ to models of $T(x,y)$.
If such an extension is possible, then the input formula is valid, and a relationship between $x$ and $y$ are gathered in a Skolem function.
Otherwise the formula is invalid, and no Skolem function exists.
We refer the reader to~\cite{KatisFGBGW16} for more details on the Skolem-function generation.

Our approach presented in this paper relies on the fact that during each run, \aeval iteratively creates a set of formulas $\{P_i(x)\}$, such that each $P_i(x)$ has a common model with $S(x)$ and $P_i(x) \Rightarrow \exists y \,.\,T (x,y)$.
After $n$ iterations, \aeval establishes a formula $R_n(x) \eqdef \bigvee_{i=1}^n P_i(x)$ which by construction implies $\exists y\,.\,T(x,y)$.
If additionally $S(x)\Rightarrow R_n(x)$, the input formula is valid, and the algorithm terminates.
Fig.~\ref{fg:aeval} shows a Venn diagram for an example of the opposite scenario: $R_2(x) = T_1(x) \lor T_2(x)$, but the input formula is invalid.
However, models of each $S(x) \land P_i(x)$ can still be extended to a model of $T(x, y)$.

In general, if after $n$ iterations $S(x) \land T(x,y) \land \neg R_n(x)$ is unsatisfiable,
then \aeval terminates.
Note that the formula $\forall x.~ S(x) \land R_n(x) \Rightarrow \exists y .~T(x,y)$ is valid by construction at any iteration of the algorithm.
%
We say that $R_n(x)$ is a \emph{region of validity}, and in this work, we are interested in the \emph{maximal} regions of validity, i.e., the ones produced by disjoining all $\{P_i(x)\}$ produced by \aeval before termination and by conjoining it with $S(x)$.
Throughout the paper, we assume that all regions of validity are maximal.



\begin{figure}[!t]
\centering
\includegraphics[scale=0.47]{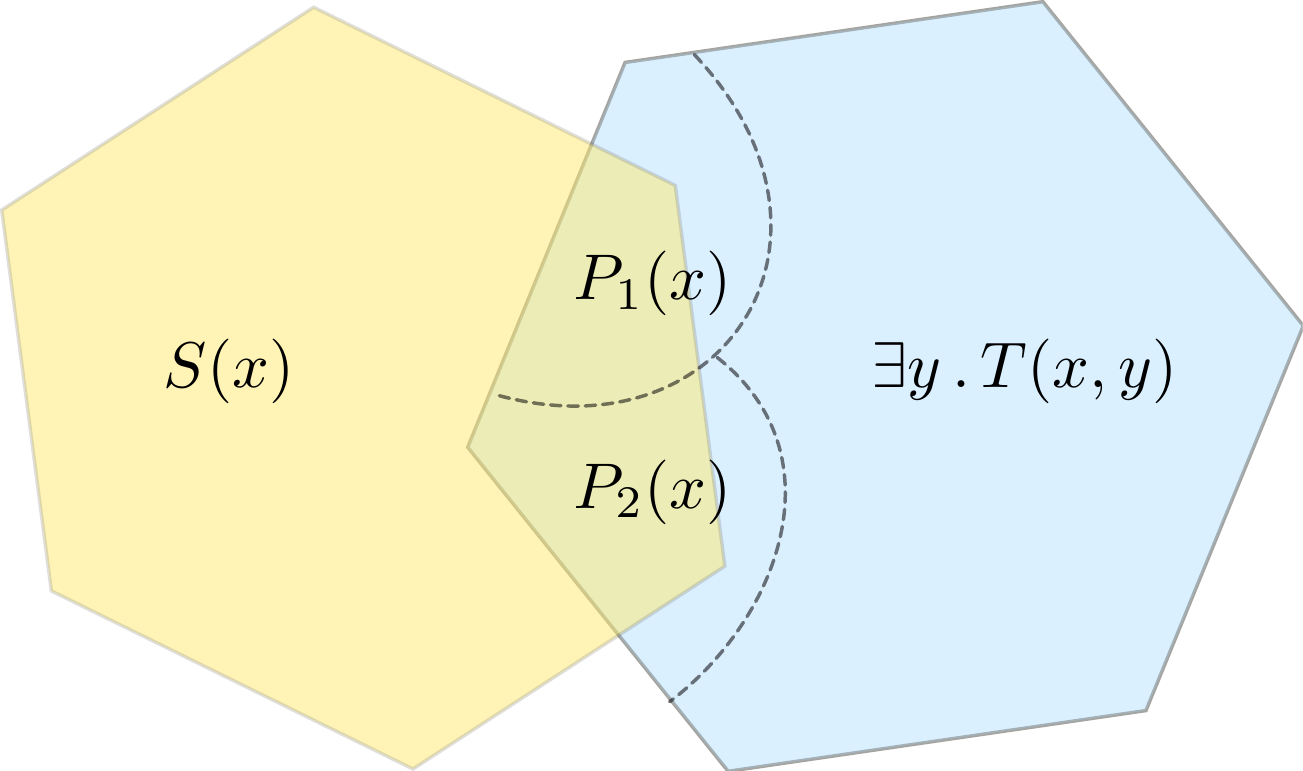}
\caption{Region of validity computed for an example requiring \aeval to iterate two times.}
\label{fg:aeval}
\end{figure}

\begin{lemma}\label{lem:aeval}
  Let $R_n(x)$ be the region of validity returned by \aeval for  formula $\forall
  s.~ S(x) \Rightarrow \exists y\,.\,T(x,y)$. Then
$  \forall x.~ S(x) \Rightarrow (R_n(x) \Leftrightarrow \exists y\,.\,T(x,y))$.
\end{lemma}
\begin{proof}
  ($\Rightarrow$) By construction of $R_n(x)$.

  ($\Leftarrow$) Suppose towards contradiction that the formula does
  not hold. Then there exists $x_0$ such that $S(x_0) \land (\exists
  y. T(x_0, y)) \land \neg R_n(x_0)$ holds. But this is a direct
  contradiction for the termination condition for \aeval. Therefore
  the original formula does hold.
\end{proof}





\section{Validity-Guided Synthesis from Assume-Guarantee Contracts}
\label{sec:synthesis}

\begin{algorithm*}[!t]
\caption{\jsynvg (A : assumptions, G : guarantees)}
\label{alg:synthesis}
\begin{algorithmic}[1]
	\State $F(s) \gets \mathit{true}$\label{alg:init};\Comment{Fixpoint of viable states}
	\While{$\mathit{true}$}
		\State $\phi \gets \forall s,i. \ (F(s) \land A(s,i) \Rightarrow \exists s'.G_{T}(s,i,s') \land F(s'))$\label{alg:ae1};
		\State $\tuple{\mathit{valid}, \subs, \skolems} \gets \aeval(\phi) \label{alg:val1}$;
		\If{$\mathit{valid}$\label{alg:val2}}
            \If{$\exists s . G_{I}(s) \land F(s)$}
				\Return $\tuple{\realizable, \skolems, s, F}$\label{alg:issat};
			\Else{\Comment{Empty set of initial or viable states}}
		 		\Return $\unrealizable$\label{alg:unreal};
		 	\EndIf
		\Else{\Comment{Extract region of validity $Q(s,i)$}}
			\State $Q(s,i) \gets \subs$\label{alg:valreg};
			\State $\phi' \gets \forall s. \ (F(s) \Rightarrow \exists i. A(s,i) \land \lnot
			Q(s,i))$\label{alg:ae2};
			\State $\tuple{\_, \mathit{violatingRegion}, \_} \gets \aeval(\phi')$;
				\State $W(s) \gets \mathit{violatingRegion}$;
			\State $F(s) \gets F(s) \land \lnot W(s)$\label{alg:rem};\Comment{Refine set of viable states}	
			
				
		\EndIf
	\EndWhile
\end{algorithmic}
\end{algorithm*}

Alg.~\ref{alg:synthesis}, named \jsynvg (for {\em validity guided}), shows the validity-guided technique that we use towards the automatic synthesis of implementations. The specification is written using the Assume-Guarantee convention that we described in Section~\ref{sec:background} and is provided as an input.
The algorithm relies on \aeval, for each call of which we write $\tuple{x, y, z} \gets \aeval(\ldots)$: $x$ specifies if the given formula is $\mathit{valid}$ or $\mathit{invalid}$, $y$ identifies the region of validity (in both cases), and $z$ -- the Skolem function (only in case of the validity).

The algorithm maintains a formula $F(s)$ which is initially assigned $\mathit{true}$ (line~\ref{alg:init}).
It then attempts to strengthen $F(s)$ until it only contains viable states (recall Eqs.~\ref{eq:viable}
and~\ref{eq:nonempty}), i.e., a greatest fixpoint is reached.
We first encode Eq.~\ref{eq:viable} in a formula $\phi$ and then provide it as input to \aeval (line~\ref{alg:val1}) which determines its validity (line~\ref{alg:val2}).
If the formula is valid, then a witness $\skolems$ is non-empty.
By construction, it contains valid assignments to the existentially quantified variables of $\phi$.
In the context of viability, this witness is capable of providing viable states that can be used as a safe
reaction, given an input that satisfies the assumptions.

With the valid formula $\phi$ in hand, it remains to check that the fixpoint intersects with the initial states, i.e., to find a model of formula in Eq.~\ref{eq:nonempty} by a simple satisfiability check.
If a model exists, it is directly combined with the extracted witness and used towards an implementation of the system, and the algorithm terminates (line~\ref{alg:issat}).
Otherwise, the contract is unrealizable since either there are no states that satisfy the
initial state guarantees $G_I$, or the set of viable states $F$ is empty.

If $\phi$ is not true for every possible assignment of the universally
quantified variables, \aeval provides a \textit{region of validity} $Q(s,i)$
(line~\ref{alg:valreg}).
At this point, one might assume that $Q(s,i)$ is sufficient to restrict $F$ towards a solution. This is not the case since $Q(s,i)$ creates a subregion
involving both state and input variables. As such, it may contain constraints
over the contract's inputs above what are required by $A$, ultimately leading to implementations that only work correctly for a small part of the input domain.

Fortunately, we can again use \aeval's capability of providing regions of validity
towards removing inputs from $Q$.  Essentially, we want to remove those states from $Q$ if even one input causes them to violate the formula on line~\ref{alg:ae1}.  We denote by $W$ the {\em violating region} of $Q$.  To construct $W$, \aeval  determines
the validity of formula $\phi' \gets \forall s. \ (F(s) \Rightarrow \exists
i. A(s,i) \land \lnot Q(s,i))$ (line~\ref{alg:ae2}) and computes
a new region of validity.

If $\phi'$ is invalid, it indicates that there are still non-violating states (i.e., outside $W$) that may lead to a fixpoint.
Thus, the algorithm removes the unsafe states from $F(s)$ in line~\ref{alg:rem}, and iterates until a greatest fixpoint for $F(s)$ is reached.  
If $\phi'$ is valid, then every state in $F(s)$ is unsafe, under a specific input that satisfies the contract assumptions (since $\lnot Q(s,i)$ holds in this case), and the specification is unrealizable (i.e., in the next iteration, the algorithm will reach line~\ref{alg:unreal}).

\subsection{Soundness}
\label{sec:soundness}

\begin{lemma}
  $\viable \Rightarrow F$ is an invariant for
  Alg.~\ref{alg:synthesis}.
\label{lem:alg1-viable}
\end{lemma}

\begin{proof}
  It suffices to show this invariant holds each time $F$ is assigned.
  On line~\ref{alg:init}, this is trivial. For line~\ref{alg:rem}, we can assume that
  $\viable \Rightarrow F$ holds prior to this line. Suppose towards
  contradiction that the assignment on line~\ref{alg:rem} violates the invariant.
  Then there exists $s_0$ such that $F(s_0)$, $W(s_0)$, and
  $\viable(s_0)$ all hold. Since $W$ is the region of validity for
  $\phi'$ on line~\ref{alg:ae2}, we have
  $W(s_0) \land F(s_0) \Rightarrow \exists i. A(s_0, i) \land \neg Q(s_0, i)$
  by Lemma~\ref{lem:aeval}. Given that $W(s_0)$ and $F(s_0)$ hold, let $i_0$
  be such that $A(s_0, i_0)$ and $\neg Q(s_0, i_0)$ hold. Since $Q$ is the
  region of validity for $\phi$ on line~\ref{alg:ae1}, we have
  $F(s_0) \land A(s_0, i_0) \land \exists s'. G_T(s_0, i_0, s') \land F(s') \Rightarrow Q(s_0, i_0)$
  by Lemma~\ref{lem:aeval}.
  Since $F(s_0)$, $A(s_0, i_0)$ and $\neg Q(s_0, i_0)$ hold, we conclude that
  $\exists s'. G_T(s_0, i_0, s') \land F(s') \Rightarrow \bot$.
  We know that
  $\viable \Rightarrow F$ holds prior to line~\ref{alg:rem}, thus
  $\exists s'. G_T(s_0, i_0, s') \land \viable(s')\Rightarrow \bot$. But this is a
  contradiction since $\viable(s_0)$ holds. Therefore the invariant holds on
  line~\ref{alg:rem}.
\end{proof}

\begin{theorem}
  The \realizable and \unrealizable results of
  Alg.~\ref{alg:synthesis} are sound.
\end{theorem}


\begin{proof}
If Alg.~\ref{alg:synthesis} terminates, then the
formula for $\phi$ on line~\ref{alg:ae1} is valid. Rewritten, $F$
satisfies the formula
\begin{equation}
  \forall s.~F(s) \Rightarrow \left(\forall i.~ A(s,i) \Rightarrow \exists
    s'.G_{T}(s,i,s') \land F(s')\right).
  \label{eq:F-rewritten}
\end{equation}
Let the function $f$ be defined over state predicates as
  \begin{equation}
    f = \lambda V. \lambda s.~ \forall i.~ A(s,i) \Rightarrow \exists s'.G_{T}(s,i,s') \land V(s').
    \label{eq:f-fixed-point}
  \end{equation}
  State predicates are equivalent to subsets of the state space and
  form a lattice in the natural way. Moreover, $f$ is monotone on this
  lattice. From Eq.~\ref{eq:F-rewritten} we have
  $F \Rightarrow f(F)$. Thus $F$ is a post-fixed point of $f$. In
  Eq.~\ref{eq:viable}, $\viable$ is defined as the greatest
  fixed-point of $f$. Thus $f \Rightarrow \viable$ by the Knaster-Tarski
  theorem. Combining this with Lemma~\ref{lem:alg1-viable}, we have
  $F = \viable$. Therefore the check on line~\ref{alg:issat} is equivalent to the
  check in Eq.~\ref{eq:nonempty} for realizability.
\end{proof}

\subsection{Termination on finite models}
\label{sec:termfinal}
\begin{lemma}
Every loop iteration in Alg.~\ref{alg:synthesis} either
terminates or removes at least one state from $F$.
\label{lem:progress}
\end{lemma}
\begin{proof}
  It suffices to show that at least one state is removed from $F$ on
  line~\ref{alg:rem}. That is, we want to show that $F \cap W \neq \varnothing$ since
  this intersection is what is removed from $F$ by line~\ref{alg:rem}. 
  
  If the query on line~\ref{alg:val1} is valid, then the algorithm terminates.   If not, then there exists a state $s^{*}$ and input $i^{*}$ such that $F(s^{*})$ and $A(s^{*}, i^{*})$ such that there is no state $s'$ where both $G(s^{*}, i^{*}, s')$ and $F(s')$ hold.  Thus, $\lnot Q(s^{*}, i^{*})$, and $s^{*} \in \mathit{violatingRegion}$, so $W \neq \varnothing$.  Next, suppose towards contradiction that $F \cap W = \varnothing$ and $W \neq \varnothing$. Since $W$ is the
  region of validity for $\phi'$ on line~\ref{alg:ae2}, we know that $F$ lies
  completely outside the region of validity and therefore
  $\forall s.~ \neg \exists i. A(s,i) \land \neg Q(s, i)$
  by Lemma~\ref{lem:aeval}. Rewritten,
  $\forall s, i.~ A(s, i) \Rightarrow Q(s, i)$. Note that $Q$ is the
  region of validity for $\phi$ on line~\ref{alg:ae1}. Thus $A$ is completely
  contained within the region of validity and formula $\phi$ is valid.
  This is a contradiction since if $\phi$ is valid then line~\ref{alg:rem} will
  not be executed in this iteration of the loop. Therefore
  $F \cap W \neq \varnothing$ and at least one state is removed from $F$
  on line~\ref{alg:rem}.
\end{proof}

\begin{theorem}
For finite models, Alg.~\ref{alg:synthesis} terminates.
\end{theorem}
\begin{proof}
Immediately from Lemma~\ref{lem:progress} and the fact that \aeval terminates on finite models~\cite{fedyukovich2015automated}.
\end{proof}

\begin{figure}[!t]
\centering
 \begin{Verbatim}[fontsize=\scriptsize]
const C = 2.0;

-- empty buckets e and e+1 each round
node game(i1,i2,i3,i4,i5: real; e: int) returns (guarantee: bool);
var
  b1, b2, b3, b4, b5 : real;
let
  assert i1 >= 0.0 and i2 >= 0.0 and i3 >= 0.0 and i4 >= 0.0 and i5 >= 0.0;
  assert i1 + i2 + i3 + i4 + i5 = 1.0;

  b1 = 0.0 -> (if (e = 5 or e = 1) then i1 else (pre(b1) + i1));
  b2 = 0.0 -> (if (e = 1 or e = 2) then i2 else (pre(b2) + i2));
  b3 = 0.0 -> (if (e = 2 or e = 3) then i3 else (pre(b3) + i3));
  b4 = 0.0 -> (if (e = 3 or e = 4) then i4 else (pre(b4) + i4));
  b5 = 0.0 -> (if (e = 4 or e = 5) then i5 else (pre(b5) + i5));

  guarantee = b1 <= C and b2 <= C and b3 <= C and b4 <= C and b5 <= C;

  --%REALIZABLE i1, i2, i3, i4, i5;
  --%PROPERTY guarantee;
tel;
 \end{Verbatim}
\vspace{-1em}
\caption{An Assume-Guarantee contract for the Cinderella-Stepmother game in Lustre.}

\label{fg:cind}
\end{figure}

\subsection{Applying \jsynvg to the Cinderella-Stepmother game}
\label{sec:algexample}

Fig.~\ref{fg:cind} shows one possible interpretation of the contract designed
for the instance of the Cinderella-Stepmother game that we introduced in Sect.~\ref{sec:example}. The contract
is expressed in Lustre~\cite{lustrev6}, a language
that has been extensively used for specification as well as implementation of
safety-critical systems, and is the kernel language in SCADE, a popular tool in
model-based development. The contract is defined as a Lustre node \texttt{game}, with a global
constant \texttt{C} denoting the bucket capacity. The node describes the game itself,
through the problem's input and output variables. The main input is Stepmother's
distribution of one unit of water over five different input variables,
\texttt{i1} to \texttt{i5}. While the node contains a sixth input argument,
namely \texttt{e}, this is in fact used as the output of the system that we want to
implement, representing Cinderella's choice at each of her turns.

We specify the system's inputs \texttt{i1}, \ldots, \texttt{i5} using the \texttt{REALIZABLE} statement and define the contract's assumptions over them: $A(i_1, \ldots, i_5) = (\bigwedge_{k=1}^{5} i_k >= 0.0) \land (\sum_{k=1}^{5} i_{k} = 1.0)$. The assignment to boolean variable \texttt{guarantee} (distinguished via the \texttt{PROPERTY} statement) imposes the guarantee constraints on the buckets' states through the entire
duration of the game, using the local variables \texttt{b1} to \texttt{b5}.
Initially, each bucket is empty, and with each transition to a new state, the contents depend on
whether Cinderella chose the specific bucket, or an adjacent one. If so, the value of each \texttt{b}$_k$ at the the next turn becomes equal to the value of the corresponding input variable \texttt{i}$_k$. Formally, for the initial state, $G_{I}(C, b_1, \ldots, b_5) = (\bigwedge_{k=1}^{5} b_k = 0.0) \land (\bigwedge_{k = 1}^{5} b_k \le C)$, while the transitional guarantee is $G_T([C,b_1, \ldots, b_5, e], i_1, \ldots, i_5, [C',b_{1}', \ldots, b_{5}',e']) = (\bigwedge_{k=1}^{5} b_{k}' = ite(e = k \lor e = k_{prev}, i_k, b_k + i_k) \land (\bigwedge_{k=1}^{5} b_{k}' \le C')$, where $k_{prev} = 5$ if $k = 1$, and $k_{prev} = k - 1$ otherwise. Interestingly, the lack of explicit constraints over $e$, i.e. Cinderella's choice, permits the action of Cinderella skipping her current turn, i.e. she does not choose to empty any of the buckets. With the addition of the guarantee $(e = 1) \lor \ldots \lor (e =5)$, the contract is still realizable, and the implementation is verifiable, but Cinderella is not allowed to skip her turn anymore.

If the bucket was not covered by Cinderella's choice, then its contents are
updated by adding Stepmother's distribution to the volume of water that the
bucket already had. The arrow (\texttt{->}) operator distinguishes the initial state (on the left) from subsequent states (on the right), and variable values in the previous state can be accessed using the \texttt{pre} operator.
The contract should only be realizable if, assuming valid inputs given by the Stepmother
(i.e. positive values to input variables that add up to one water unit),
Cinderella can keep reacting indefinitely, by providing outputs that satisfy the
guarantees (i.e. she empties buckets in order to prevent overflow in Stepmother's next turn).
We provide the contract in Fig.~\ref{fg:cind} as input to  Alg.~\ref{alg:synthesis} which then iteratively attempts to construct a fixpoint of viable states, closed under the transition relation.

Initially $F = \mathit{true}$, and we query \aeval for the validity of formula $\forall i_1, \ldots,$ $i_5,b_1, \ldots, b_5 \,.\, A(i_1, \ldots, i_5) \Rightarrow \exists b'_1, \ldots, b'_5, e\,.\,G_{T}(i_1, \ldots, i_5,b_1, \ldots, b_5, b'_1,$ $\ldots, b'_5, e)$.
Since $F$ is empty, there are states satisfying $A$, for which there is no transition to $G_{T}$.
In particular, one such counterexample identified by \aeval is represented by the set of assignments $\mathit{cex} = \{\ldots,b_{4} = 3025, i_{4} = 0.2, b'_{4} = 3025.2, \ldots\}$, where the already overflown bucket $b_4$ receives additional water during the transition to the next state, violating the contract guarantees.
In addition, \aeval provides us with a region of validity $Q(i_1, \ldots, i_5,b_1, \ldots, b_5)$, a formula for 
which $\forall i_1, \ldots,$ $i_5,b_1, \ldots, b_5 \,.\, A(i_1, \ldots, i_5) \land Q(i_1, \ldots, i_5,b_1, \ldots, b_5) \Rightarrow \exists b'_1, \ldots, b'_5, e\,.\,G_{T}(i_1, \ldots, i_5,b_1, \ldots, b_5, b'_1,$ $\ldots, b'_5, e)$ is valid.
Precise encoding of $Q$ is too large to be presented in the paper; intuitively it contains some constraints on $i_1,\ldots,i_5$ and $b_1,\ldots,b_k$ which are stronger than $A$ and which block the inclusion of violating states such as the one described by $\mathit{cex}$.

Since $Q$ is defined over both state and
input variables, it might contain constraints over the inputs, which is an
undesirable side-effect.
%
In the next step, \aeval decides the validity of formula $\forall b_1, \ldots, b_5\,.\, \exists
i_1, \ldots, i_5 \,.\, A(i_1, \ldots, i_5) \land \lnot Q(i_1, \ldots, i_5,b_1, \ldots, b_5)$ and extracts a violating region $W$ over $b_1, \ldots, b_5$.
%
Precise encoding of $W$ is also too large to be presented in the paper; and intuitively it captures certain steps in which Cinderella may not take the optimal action. 
Blocking them leads us eventually to proving the contract's realizability. 

From this point on, the algorithm continues following the steps
explained above. 
In particular, it terminates after one
more refinement, at depth 2. At that point, the refined version of
$\phi$ is valid, and \aeval constructs a witness containing valid reactions to
environment behavior.
In general, the witness is described through the use of nested \textit{if-then-else} blocks, where the conditions are subsets of the antecedent of the implication in formula $\phi$, while the body contains valid assignments to
state variables to the corresponding subset.

\section{Implementation and Evaluation}
\label{sec:impl}

The implementation of the algorithm has been added to a branch of the  \jkind~\cite{jkind} model checker\footnote{The \jkind fork with \jsynvg is available at \url{https://goo.gl/WxupTe}.}.  \jkind officially supports synthesis
using a $k$-inductive approach, named \jsyn~\cite{KatisFGBGW16}. For clarity, we named
our validity-guided technique \jsynvg (i.e., validity-guided synthesis). \jkind uses Lustre~\cite{lustrev6} as its specification and implementation language.
%
\jsynvg encodes Lustre specifications in the language of
linear real and integer arithmetic (LIRA)
and communicates them to \aeval\footnote{The \aeval tool is available at~\url{https://goo.gl/CbNMVN}.}.
%
%
Skolem functions returned by \aeval get then translated 
into an efficient and practical implementation. To compare the quality of implementations against \jsyn, we use
\smtlibtoc, a tool that has been specifically developed to translate
  Skolem functions to C implementations\footnote{The \smtlibtoc tool is available at \url{https://goo.gl/EvNrAU}.}.


\subsection{Experimental results}
\label{sec:results}

We evaluated \jsynvg by synthesizing implementations
for 124 contracts%
\footnote{All of the benchmark contracts can be found at
\url{https://goo.gl/2p4sT9}.}
originated from a broad variety of contexts. Since we have been unable to find past work that contained benchmarks directly relevant to our approach, we propose a comprehensive collection of contracts that can be used by the research community for future advancements in reactive system synthesis for contracts that rely on infinite theories. Our benchmarks are split into three categories:
\begin{itemize}
\item 59 contracts correspond to various industrial projects, such as a Quad-Redundant Flight Control System, a Generic Patient Controlled Analgesia infusion pump, as well as a collection of contracts
for a Microwave model, written by graduate students as part of a software
engineering class;
\item 54 contracts were initially used for the verification of existing handwritten implementations~\cite{hagen2008scaling};
\item 11 models contain variations of the
Cinderella-Stepmother game, as well as examples that we created.
\end{itemize}
All of the synthesized implementations were verified against the original contracts using \jkind.

The goal of this experiment was to determine the performance and generality of the \jsynvg algorithm.  We compared against the existing \jsyn algorithm, and for the Cinderella model, we compared against~\cite{beyene2014constraint} (this was the only synthesis problem in the paper).  We examined the following aspects:
\begin{itemize}
    \item time required to synthesize an implementation;
    \item size of generated implementations in lines of code (LoC);
    \item execution speed of generated C implementations derived from the synthesis procedure; and
    \item number of contracts that could be synthesized by each approach.
\end{itemize}
\noindent Since \jkind already supports synthesis through \jsyn, we were able to directly
compare \jsynvg against \jsyn's $k$-inductive approach. We
ran the experiments using a computer with Intel Core i3-4010U 1.70GHz CPU and
16GB RAM.

A listing of the statistics that we tracked while running experiments is
presented in Table~\ref{tbl:stats}.
Fig.~\ref{fg:performance} shows the time allocated by \jsyn and \jsynvg to solve each problem, with \jsynvg
outperforming \jsyn for the vast majority of the benchmark suite, often times by a margin greater than
50\%. Fig.~\ref{fg:size} on the other hand, depicts small differences in the
overall size between the synthesized implementations. While it would be
reasonable to conclude that there are no noticeable improvements, the big picture is different:
solutions by \jsynvg always require just a single Skolem function, but solutions by \jsyn may require several ($k-1$ to initialize the system, and one for the inductive step).
In our evaluation, \jsyn proved the realizability of the majority of benchmarks by constructing proofs of length $k=0$, which essentially means that the entire space of states is an inductive invariant. 
However, several spikes in Fig.~\ref{fg:size} refer to benchmarks, for which \jsyn constructed a proof of length $k>0$, which was significantly longer that the corresponding proof by \jsynvg.
Interetsingly, we also noticed cases where \jsyn
implementations are (insignificantly) shorter. This provides us with another 
observation regarding the formulation of the problem for $k=0$ proofs. In
these cases, \jsyn proves the existence of viable states, starting from a set
of \textit{pre-initial} states, where the contract does not need to hold. This
has direct implications to the way that the $\forall\exists$-formulas are
constructed in \jsyn's underlying machinery, where the assumptions are ``baked''
into the transition relation, affecting thus the performance of \aeval.

\begin{table}[!t]
\centering
\begin{minipage}{\textwidth}
\centering
\caption{Benchmark statistics.}
\vspace{-1em}
\label{tbl:stats}
\resizebox{0.72\textwidth}{!}{%
\begin{tabular}{|l|c|c|}
\hline
 & \jsyn & \jsynvg \\ \hline 
Problems solved & 113 & \textbf{124} \\ \hline
Performance (avg - seconds) & 5.72 & \textbf{2.78} \\ \hline
Performance (max - seconds) & 352.1 & \textbf{167.55} \\ \hline
Implementation Size (avg - Lines of Code) & 72.88 & \textbf{70.66} \\ \hline
Implementation Size (max - Lines of Code) & 2322 & \textbf{2142} \\ \hline
Implementation Performance (avg - ms) & 57.84 & \textbf{56.32} \\ \hline
Implementation Performance (max - ms) & 485.88 & \textbf{459.95} \\ \hline
\end{tabular}
}
\end{minipage}
\begin{minipage}{\textwidth}
\centering
\caption{Cinderella-Stepmother results.}
\vspace{-1em}
\label{tbl:cindtbl}
\resizebox{0.8\textwidth}{!}{%
\begin{tabular}{|l|c|c|c|c|c|}
\hline
\multirow{2}{*}{Game} & \multicolumn{3}{c|}{\jsynvg} & \multicolumn{2}{c|}{\textsc{ConSynth}~\cite{beyene2014constraint}} \\ \cline{2-6}
 & \begin{tabular}[c]{@{}c@{}}Impl. Size\\ (LoC)\end{tabular} & \begin{tabular}[c]{@{}c@{}}Impl. Performance\\ (ms)\end{tabular} & Time & \begin{tabular}[c]{@{}c@{}}Time\\ (Z3)\end{tabular} & \begin{tabular}[c]{@{}c@{}}Time\\ (Barcelogic)\end{tabular} \\ \hline
Cind (C = 3) & 204 & 128.09 & 4.5s & \multirow{2}{*}{3.2s} & \multirow{2}{*}{1.2s} \\ \cline{1-4}
Cind2 (C = 3) & 2081 & 160.87 & 28.7s &  &  \\ \hline
Cind (C = 2) & 202 & 133.04 & 4.7s & \multirow{2}{*}{1m52s} & \multirow{2}{*}{1m52s} \\ \cline{1-4}
Cind2 (C = 2) & 1873 & 182.19 & 27.2s &  &  \\ \hline
\end{tabular}
}
\end{minipage}

\end{table}

\begin{figure}[!t]
\centering
\vspace{-5pt}
\hspace{-2em}
\subfloat[Performance of synthesizers]{
\includegraphics[width=2.15in]{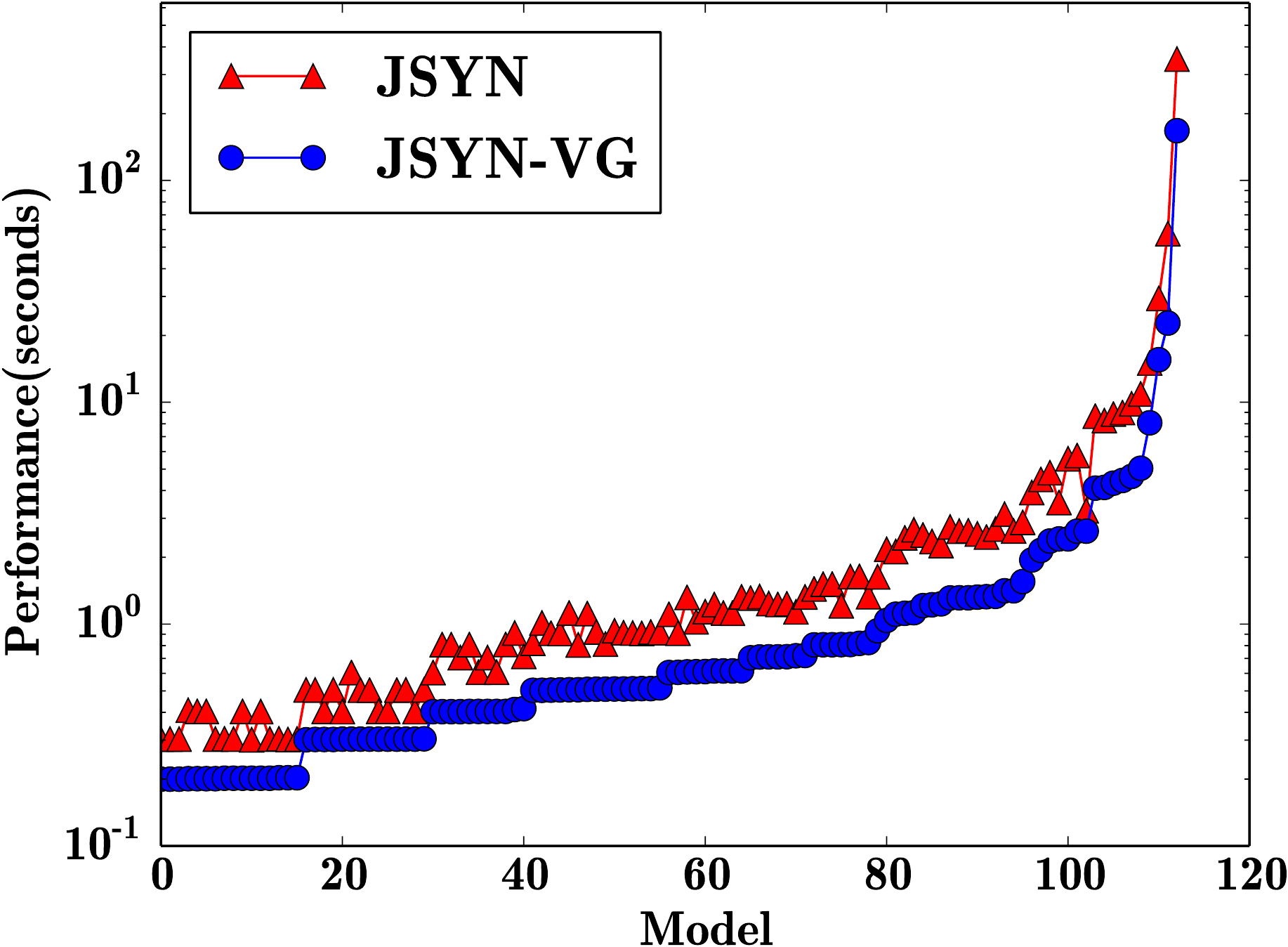}
\label{fg:performance}}
\quad
\subfloat[Size of implementations]{\includegraphics[width=2.15in]{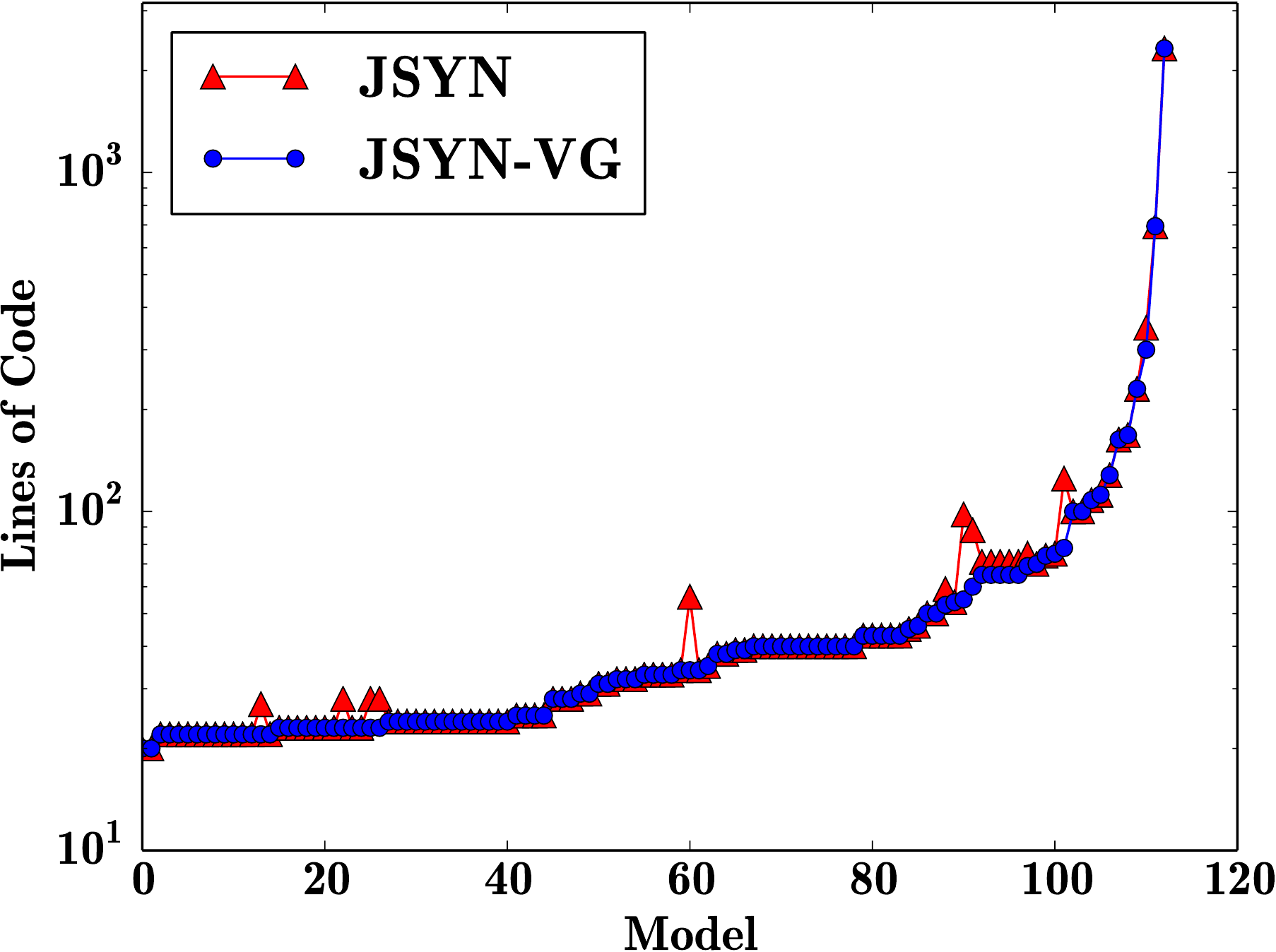}
\label{fg:size}}
\vspace{-5pt}
\subfloat[Performance of implementations]{
\hspace{-1em}
\includegraphics[width=2.15in]{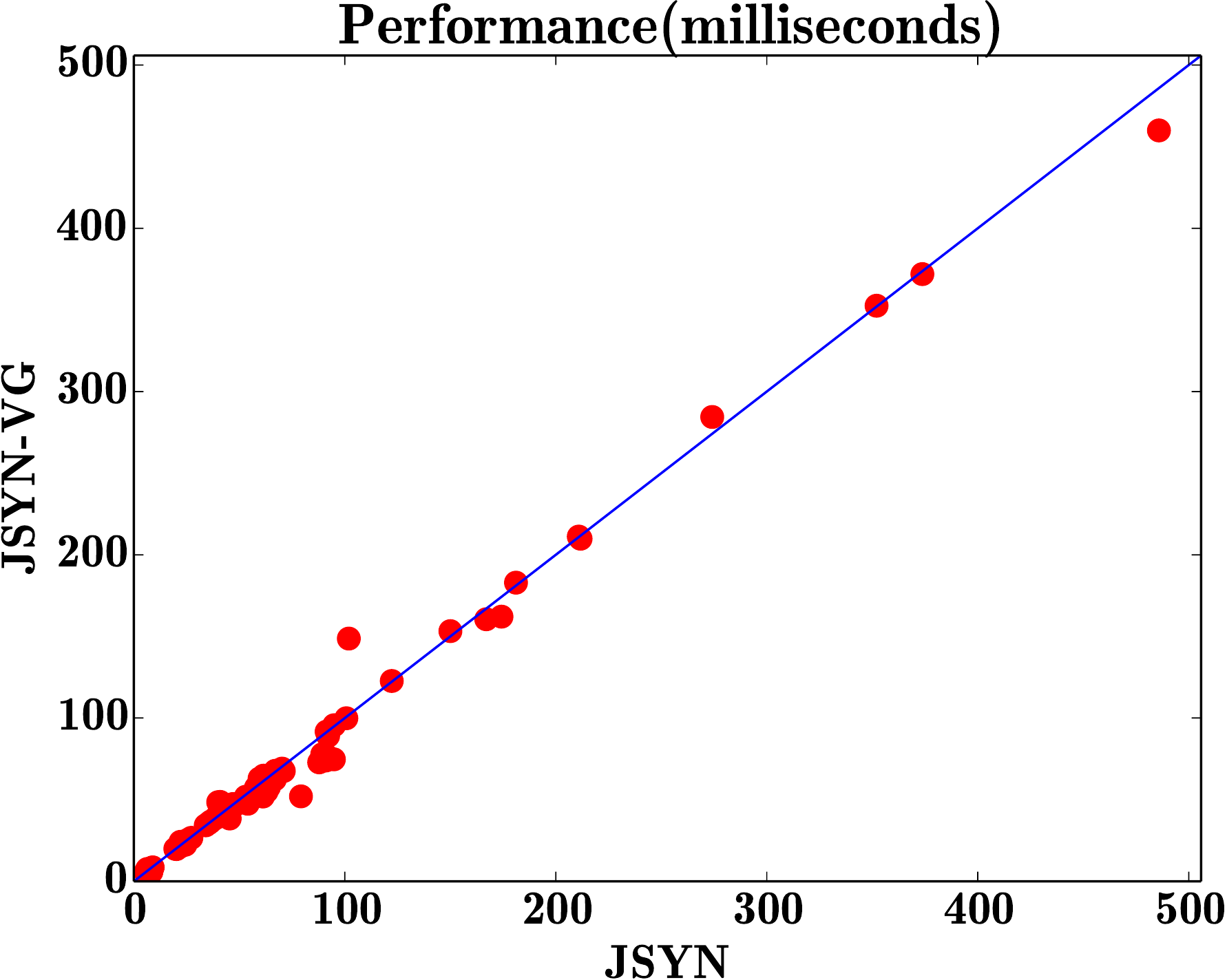}
\label{fg:implperformance}}
\caption{Experimental results.}
\vspace{-5pt}
\label{fg:results}
\end{figure}

One last statistic that we tracked was the performance of the synthesized C
implementations in terms of execution time, which can be seen in Fig.~\ref{fg:implperformance}. The performance was computed as the mean of 1000000 iterations of executing each implementation using random input values.
 According to the figure as well as Table~\ref{tbl:stats}, the differences are minuscule on average.

Fig.~\ref{fg:results} does not cover the entirety of the
benchmark suite. From the original 124 problems, eleven of them cannot be
solved by \jsyn's $k$-inductive approach.
Four of these files are variations of
the Cinderella-Stepmother game using different representations of the game, as well as two different values
for the bucket capacity (2 and 3). Using the variation in Fig.~\ref{fg:cind} as an input to \jsyn, we receive an ``unrealizable'' answer, with the counterexample shown
in Fig.~\ref{fg:cex}. Reading through the feedback provided by \jsyn, it is
apparent that the underlying SMT solver is incapable of choosing the correct
buckets to empty, leading eventually to a state where an overflow occurs for the
third bucket. As we already discussed though, a winning strategy exists for the
Cinderella game, as long as the bucket capacity \texttt{C} is between 1.5 and 3. This
provides an excellent demonstration of the inherent weakness of \jsyn
for determining unrealizability. \jsynvg's validity-guided approach,
is able to prove the realizability for these contracts, as
well as synthesize an implementation for each.

\begin{figure}[!t]
\centering
 \begin{Verbatim}[fontsize=\scriptsize]
	 ++++++++++++++++++++++++++++++++++++++++++++++++++++++++++
	      UNREALIZABLE || K = 6 || Time = 2.017s
	                 Step
	      variable      0    1      2      3      4      5
	      INPUTS
	      i1            0    0      0 0.416* 0.944* 0.666*
	      i2            1    0 0.083* 0.083*      0 0.055*
	      i3            0    1 0.305*    0.5 0.027* 0.194*
	      i4            0    0 0.611*      0      0 0.027*
	      i5            0    0      0      0 0.027* 0.055*
	
	      OUTPUTS
	      e             1    3      1      5      4      5
	
	      NODE OUTPUTS
	      guarantee   true true   true   true   true  false
	
	      NODE LOCALS
	      b1            0    0      0 0.416* 1.361* 0.666*
	      b2            0    0 0.083* 0.166* 0.166* 0.222*
	      b3            0    1 1.305* 1.805* 1.833* 2.027*
	      b4            0    0 0.611* 0.611*      0 0.027*
	      b5            0    0      0      0 0.027* 0.055*
	
	      * display value has been truncated
	 ++++++++++++++++++++++++++++++++++++++++++++++++++++++++++
 \end{Verbatim}
\vspace{-1.5em}
\caption{Spurious counterexample for Cinderella-Stepmother example using \jsyn}
\label{fg:cex}
\end{figure}

Table~\ref{tbl:cindtbl} shows how \jsynvg performed on the four contracts describing the Cinderella-Stepmother game. We used two different interpretations for the game, and exercised both for the cases where the bucket capacity  \texttt{C} is equal to 2 and 3. Regarding the synthesized implementations, their size is analogous to the complexity of the program (Cinderella2 contains more local variables and a helper function to empty buckets). Despite this, the implementation performance remains the same across all implementations. Finally for reference, the table contains the results from the template-based approach followed in \textsc{Consynth}~\cite{beyene2014constraint}. From the results, it is apparent that providing templates yields better performance for the case of $C = 3$, but our approach overperforms \textsc{Consynth} when it comes to solving the harder case of $C = 2$. Finally, the original paper for \textsc{Consynth} also explores the synthesis of winning strategies for Stepmother using the liveness property that a bucket will eventually overflow. While \jkind does not natively support liveness properties, we successfully synthesized an implementation for Stepmother using a bounded notion of liveness with counters. We leave an evaluation of this category of specifications for future work.

Overall, \jsynvg's validity-guided approach provides significant advantages
over the $k$-inductive technique followed in \jsyn, and effectively expands
\jkind's solving capabilities regarding specification realizability. On top of that, it provides an efficient ``hands-off'' approach that is capable of solving complex games.
The most significant contribution, however, is the applicability of this approach, as it is not tied to a specific environment since it can be extended to support more
theories, as well as categories of specification. 

\section{Related Work}
\label{sec:related}

The work presented in this paper is closely related to approaches that attempt
to construct infinite-state implementations. Some focus on the continuous
interaction of the user with the underlying machinery, either through the use of
templates~\cite{srivastava2013template,beyene2014constraint}, or environments where the user attempts to guide the solver by choosing reactions from a collection of different
interpretations~\cite{ryzhyk2016developing}.
In contrast, our approach is completely automatic and does not require
human ingenuity to find a solution.
Most importantly, the user
does not need to be deeply familiar with the problem at hand.

Iterative strengthening of candidate formulas is also used in abductive inference~\cite{dillig2013inductive}
of loop invariants.
Their approach generates candidate invariants as maximum universal subsets (MUS) of quantifier-free formulas of the form $\phi \Rightarrow \psi$.
While a MUS may be sufficient to prove validity, it may also mislead the invariant search, so the authors use a backtracking procedure that discovers new subsets while avoiding spurious results. By comparison, in our approach the regions of validity are maximal and therefore backtracking is not required.  More importantly, reactive synthesis requires mixed-quantifier formulas, and it requires that inputs are unconstrained (other than by the contract assumptions), so substantial modifications to the MUS algorithm would be necessary to apply the approach of~\cite{dillig2013inductive} for reactive synthesis.  

The concept of synthesizing implementations by discovering fixpoints was mostly
inspired by the IC3 / PDR~\cite{bradley2011sat,een2011efficient}, which was first introduced
in the context of verification. Work from Cimatti \textit{et al.} effectively
applied this idea for the parameter synthesis in the
\textsc{HyComp} model checker~\cite{DBLP:conf/fmcad/CimattiGMT13, cimatti2015hycomp}.
Discovering fixpoints to synthesize reactive designs was first
extensively covered by Piterman \textit{et al.}~\cite{piterman2006synthesis}
who proved that the problem can be solved in cubic time for the class of GR(1) specifications.
The algorithm requires the discovery of least fixpoints for the state variables,
each one covering a greatest fixpoint of the input variables. If the specification
is realizable, the entirety of the input space is covered by the greatest fixpoints.
In contrast, our approach computes a single greatest fixpoint over the system's outputs and avoids the partitioning of the input space.  As the tools use different notations and support different logical fragments, practical comparisons are not straightforward, and thus are left for the future.

More recently, Preiner \textit{et al}. presented work on model synthesis~\cite{preiner2017counterexample}, 
that employs a counterexample-guided refinement process~\cite{reynolds2015counterexample}
to construct and check candidate models.
Internally, it relies on enumerative learning, a syntax-based technique that enumerates expressions, checks their validity against
ground test cases, and proceeds to generalize the expressions by constructing larger ones.
In contrast, our approach is syntax-insensitive in terms of generating regions of validity.
In general, enumeration techniques such as the one used in \textsc{ConSynth}'s underlying E-HSF engine~\cite{beyene2014constraint} is not an optimal strategy for our class of problems, since the witnesses constructed for the most complex contracts are described by nested if-then-else expressions of depth (i.e. number of branches) 10-20, a point at which space explosion is difficult to handle since the number of candidate solutions is large.

\section{Conclusion and Future Work}
\label{sec:conclusion}


We presented a novel and elegant approach towards the synthesis
of reactive systems, using only the knowledge provided by the
system specification expressed in infinite theories.
The main goal is to converge to a fixpoint by iteratively blocking subsets of
unsafe states from the problem space. This is achieved through the continuous
extraction of regions of validity which hint towards subsets of states that
lead to a candidate implementation.

This is the first complete attempt, to the best of our knowledge, on handling
valid subsets of a $\forall\exists$-formula to construct a greatest fixpoint on specifications expressed using infinite theories. We were able to
prove its effectiveness in practice, by comparing it to an already existing
approach that focuses on constructing $k$-inductive proofs of realizability. We showed how the new algorithm performs better than the $k$-inductive approach, both in terms of performance as well as the soundness of results. In the future, we would like to extend the applicability of this algorithm to other areas in formal verification, such as invariant generation. Another interesting goal is to make the proposed benchmark collection available to competitions such as SYNTCOMP, by establishing a formal extension for the TLSF format to support infinite-state problems~\cite{DBLP:journals/corr/Jacobs016}. Finally, a particularly interesting challenge is that of mapping infinite theories to finite counterparts, enabling the synthesis of secure and safe implementations.

\section{Data Availability Statement}
\label{sec:artifact}

The datasets generated during and/or analyzed during the current study are available in the figshare repository: \url{https://doi.org/10.6084/m9.figshare.5904904}~\cite{katis2018tacasartifact}.

\bibliography{document}
\bibliographystyle{splncs03}
\end{document}